\newif\ifrelsize

\ifrelsize
\documentclass[twoside,leqno,onecolumn]{article}
\else
\documentclass[twoside,leqno,twocolumn]{article}
\fi

\usepackage[letterpaper]{geometry}

\usepackage{ltexpprt}

\usepackage{amsmath,amsfonts,amssymb,color,bbm}
\usepackage{fancybox}
\usepackage{algorithm}
\usepackage[noend]{algpseudocode}
\usepackage{dsfont}
\usepackage{bm}
\usepackage{subfig}
\usepackage{cleveref}

\newcommand{\set}[1]{\left\{#1\right\}}

\newtheorem{definition}{Definition}

\newenvironment{fminipage}%
  {\begin{Sbox}\begin{minipage}}%
  {\end{minipage}\end{Sbox}\fbox{\TheSbox}}

\def\expec#1#2{{\mathbb{E}}_{#1}\left[ #2 \right]}

\def\defeq{\stackrel{\mathrm{def}}{=}}

\usepackage{pgf,tikz}
\usetikzlibrary{arrows}

\newcommand\Alg{\text{ALG}}
\newcommand\Opt{\text{OPT}}

\usepackage{tkz-graph}

\tikzset{
  LabelStyle/.style = { rectangle, rounded corners, draw,
                        minimum width = 2em, fill = yellow!50,
                        text = red, font = \bfseries },
  VertexStyle/.append style = { font = \bfseries},
  EdgeStyle/.append style = {->, bend right} 
  }
  
\usepackage{thmtools}
\usepackage{thm-restate}

\ifrelsize
\usepackage{relsize}
\fi

\begin{document}

\title{\Large The Traveling Firefighter Problem}

\author{Majid Farhadi\thanks{
Supported in part by \texttt{aco.gatech.edu, triad.gatech.edu}.
Part of this work was conducted when authors visited Simons Institute for the Theory of Computing.
}
 \\ Georgia Tech \\ \texttt{farhadi@gatech.edu}
\and
Alejandro Toriello \\ Georgia Tech \\ \texttt{atoriello@isye.gatech.edu}
\and
Prasad Tetali\thanks{Research supported in part by the NSF grant DMS-2055022.
The author's new affiliation is Department of Mathematical Sciences, Carnegie Mellon University, Pittsburgh.} \\ Georgia Tech \\ \texttt{tetali@math.gatech.edu}
}

\date{}

\maketitle


\fancyfoot[R]{\scriptsize{Copyright \textcopyright\ 2021 by SIAM\\
Unauthorized reproduction of this article is prohibited}}





\begin{abstract}
\ifrelsize
\relsize{3}
\fi
We introduce the $L_p$ Traveling Salesman Problem ($L_p$-TSP), given by an origin, a set of destinations, and underlying distances. The objective is to schedule a destination visit sequence for a traveler of unit speed to minimize the Minkowski $p$-norm of the resulting vector of visit/service times. For $p = \infty$ the problem becomes a path variant of the TSP, and for $p = 1$ it defines the Traveling Repairman Problem (TRP), both at the center of classical combinatorial optimization.

\smallskip

$L_p$-TSP can be formulated as a convex mixed-integer program and enables a smooth interpolation between path-TSP and TRP, corresponding to optimal routes from the perspective of a server versus the customers, respectively. The parameter $p$ can affect fairness or efficiency of the solution: The case $p = 2$, which we term the Traveling Firefighter Problem (TFP), models the scenario when the cost/damage due to a delay in service is quadratic in time. 

\smallskip

We provide a polynomial-time reduction of $L_p$-TSP (losing a factor of $1+\varepsilon$ in performance) to the segmented-TSP, a routing problem that defines a constant $O(1+\varepsilon^{-2})$ number of deadlines by which given numbers of vertices should be visited. Subsequently we derive polynomial-time approximation schemes for $L_p$-TSP in the Euclidean metric and the tree metric (for which the problem is strongly NP-hard). 

\smallskip

We also study the all-norm-TSP, in which the objective is to find a route that is (approximately) optimal with respect to the minimization of any norm of the visit times. We improve the approximation bound for this problem to $8$, down from $16$, and further prove an impossibility for an approximation factor better than $1.78$, even in line metrics. Finally, we show the performance of our algorithm can be optimized for a specific norm, particularly yielding a $5.65$-approximation for the TFP on general metrics. 

\smallskip

We leave open several interesting directions to further develop this line of research.

\end{abstract}

\ifrelsize
\relsize{3}
\fi
\section{Introduction}

The $L_p$-TSP is a routing problem seeking to minimize the $L_p$ norm of the vector of visit/service times to a set of customer locations. It generalizes and interpolates between two well-studied problems, the path variant of the TSP and the TRP, also known as the Minimum Latency Problem, which are the two extreme cases in which one minimizes either the largest or the sum of service times.  
By assuming the server's speed is constant, we use time and distance interchangeably in the remainder of the paper.

As one motivating example, the $L_p$-TSP for $p = 2$, which we call the Traveling Firefighter Problem (TFP), abstracts a macro-scale optimal strategy for dispatching a firefighter to minimize the total damage due to fires at various locations.  
Ride-sharing is another use-case of our problem. For example, devising the return route of a school bus, should we optimize the fuel consumption (i.e.\ the driver's time en route) or the average student waiting time? Is it fairer to further penalize larger waiting times, e.g.\ minimizing the sum of squares/cubes of the waiting times? Before formulating the problems under study we introduce some notation.

\smallskip

\paragraph{Notation.}
$[z]$ denotes the set $\set{1, \cdots, z}$ for any positive integer $z$. $\tilde{O}(\cdot)$ is equivalent to $O(\cdot)$, treating $\varepsilon > 0$ as a constant. A metric over a set of nodes $V$ is a distance function $d: V \times V \rightarrow \mathbb{R}_{\ge 0}$ that satisfies symmetry, $d(x,y) = d(y,x) \quad \forall x, y \in V$, identity, $d(x,x) = 0$, and the triangle inequality,
$
d(x,y) \leq d(x,z) + d(z,y) \quad \forall x,y,z \in V\,.
$

\smallskip

\paragraph{The inputs} to the problem are the set of vertices $V$, including both the destinations and the server's starting location, $s \in V$, and the underlying metric $d(\cdot,\cdot)$ over $V$, corresponding to distances (or times) between vertex pairs.

\smallskip

\paragraph{A feasible solution} or route is a permutation $\sigma$ over $V$ that starts at the origin; $\sigma_i$ denotes the $i$\textsuperscript{th} vertex to be visited, so we always have $\sigma_1 = s$.  $\mathcal{F}$ denotes the set of all feasible solutions.
The $i$\textsuperscript{th} smallest visit time,
due to a solution $\sigma \in \mathcal{F}$, 
is denoted
$T^\sigma_i$, i.e.\
\begin{equation*}
T^\sigma_i = \begin{cases}
0 & i = 1 \\
\sum_{j = 2}^{i} d(\sigma_{j-1},\sigma_j) & i \in \{2,\cdots,n\}\,
\end{cases}
\end{equation*}
The visit time for vertex $v$ is denoted $\ell^\sigma_v$. Similarly, $\ell^\sigma_s = 0$ and $\ell^\sigma_v = \sum_{i=2}^{\sigma_i = v} d(\sigma_{i-1},\sigma_i)  \quad \forall v \ne s$.

\begin{definition}[$L_p$-TSP]
The input of the optimization problem $L_p$-TSP is a set of destinations $V$, a starting vertex $s \in V$, and a metric $d: V \times V \rightarrow \mathbb{R}_{\ge 0}$. 
The objective is to find a feasible route, $\sigma \in \mathcal{F}$, starting at $s$ and visiting all $v \in V$, that minimizes the Minkowski $p$-norm of the visit times, i.e.\ $\min_{\sigma \in \mathcal{F}} \|\ell^\sigma\|_p$, where
$$
\|\ell^\sigma\|_p := \left(\sum_{v \in V} |\ell^\sigma_v|^p\right)^{\frac{1}{p}}\,.
$$
\end{definition}

When the problem/objective is clear from context, we denote an optimal route as $\Opt$ and the answer of our algorithm as $\Alg$.

\smallskip

\paragraph{A Better Objective.}
One can verify that the objective (norm) affects various aspects of the routing problem, such as efficiency. $L_p$-TSP enables a smooth transition between two 
extreme objectives. 
For larger $p$'s the objective is strongly affected by dominating (larger) entries of the delay vector and minimizing $\|\ell\|_p$ is more to the benefit of the server. In contrast, smaller $p$'s provide a further aggregated measure of the amount of time that the customers have waited. This trade-off can also be interpreted from a fairness perspective, as increasing $p$ discourages the longest waiting time from becoming too large.

\smallskip

\paragraph{Firefighter Example.}
We further elaborate on why $p \notin \set{1, \infty}$ can be a useful objective by considering the routing of a firefighter.\footnote{
Over the past decade, and for the first time on record, the annual number of acres burned in the United States exceeded $10$ million; this occurred twice \cite{Hoover18}. During $2018$, wildfire damages in California totalled $\$150$ billion \cite{wang2021economic}. 
In July $2019$, a record $2.4$ million acres of the Amazon rainforest were torched \cite{natgeo}.
The optimal allocation, scheduling, and routing of firefighting resources may help to better address this global challenge.
} 
Consider a set of wildfires in dispersed locations, and suppose a skilled firefighter extinguishes any fire the second they arrive at a location; the firefighter must choose the order in which to visit and extinguish the fires. One possible strategy is to choose a sequence in order to finish extinguishing all fires as soon as possible, which corresponds to solving the $L_\infty$-TSP (the path-TSP). However, this may not be the best solution for the firefighter: 
$L_\infty$-TSP minimizes the latest visit time for all fires, while the cost could be affected by all visit times; thus an aggregated measure could be a better objective.

For example, consider the following simple dynamics for the spread of wildfires over uniform territory, ignoring possible differences such as vegetation, weather and wind. After every second, a flammable point of territory is ignited if it is within a unit distance from the burning flames. Under these dynamics, the area of land scorched by the fire is a quadratic function of the elapsed time. 
Therefore, the damage due to the delay on the $i$\textsuperscript{th} visit can be better represented by ${\ell}_{\sigma_i}^2$ and minimizing $\sum_v \ell_v^2 = \|\ell\|_2^2$, or equivalently $\|\ell\|_2$, is a better objective for this scenario compared to other norms, particularly $\|\ell\|_1$ or $\|\ell\|_\infty$. Motivated by this example, we term $L_2$-TSP as the Traveling Firefighter Problem (TFP).

In an applied setting, this approach may require some refinement but the basic idea still applies. 
Land and weather asymmetries can be modeled by a multiplicity of vertices: If a fire spreads twice as fast in area, we can represent it by two vertices overlapping in the metric. One could also generalize the objective to a weighted sum of the squared delays and/or discretize large fires into smaller ones. Moreover, the time required to extinguish a fire can be accounted for by adding a new edge, hanging from the original destination at a distance proportional to the time required to contain that fire, and moving the destination to the other endpoint of the new edge.

\begin{figure}[t]
    \centering
    \subfloat[$L_2$-TSP route]{{
    \begin{tikzpicture}
  \SetGraphUnit{2}
  \Vertex{S}
  \WE(S){A}
  \EA(S){B}
  \EA(B){C}
  \tikzset{EdgeStyle/.append style = {bend left}}
  \Edge(S)(A)
  \Edge(A)(B)
  \Edge(B)(C)
\end{tikzpicture}
    }}%
    \qquad
    \subfloat[$L_1$-TSP route]{{
    \begin{tikzpicture}
  \SetGraphUnit{2}
  \Vertex{S}
  \WE(S){A}
  \EA(S){B}
  \EA(B){C}
  \Edge(S)(B)
  \Edge(B)(C)
  \Edge(C)(A)
\end{tikzpicture}
    }}%
    \caption{Different norms cause different routes.}%
    \label{fig:L2vsL1TSP}
\end{figure}
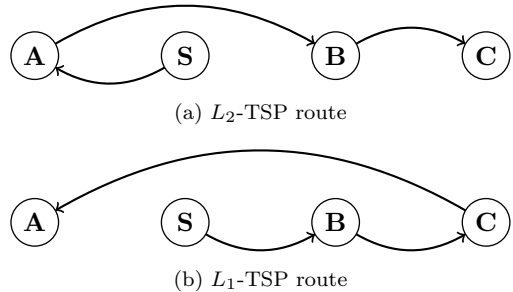

\begin{figure*}[t]
    \centering
    \subfloat[Path-TSP route]{{\includegraphics[trim={15cm 4.5cm 12cm 8cm},clip,width=7cm]{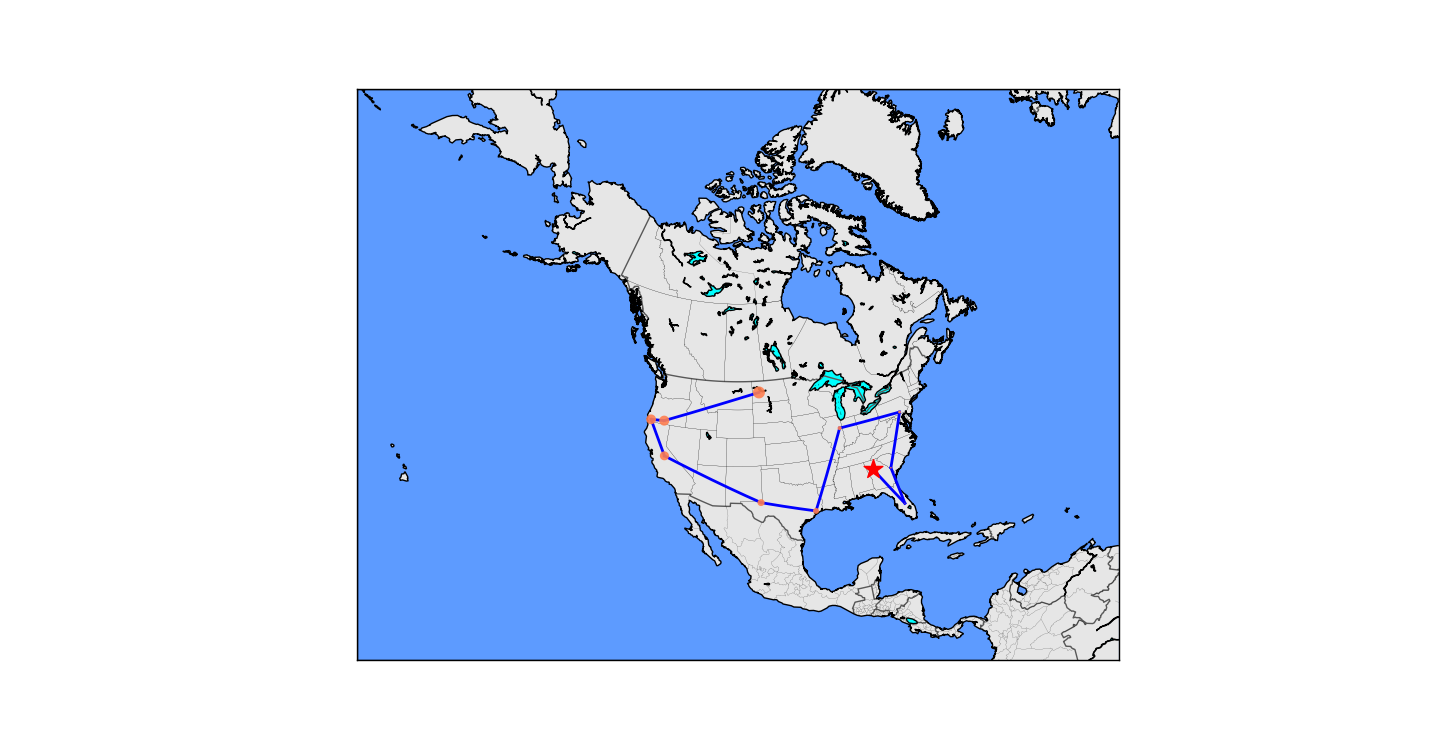}}}%
    \qquad
    \subfloat[TFP route]{{\includegraphics[trim={15cm 4.5cm 12cm 8cm},clip,width=7cm]{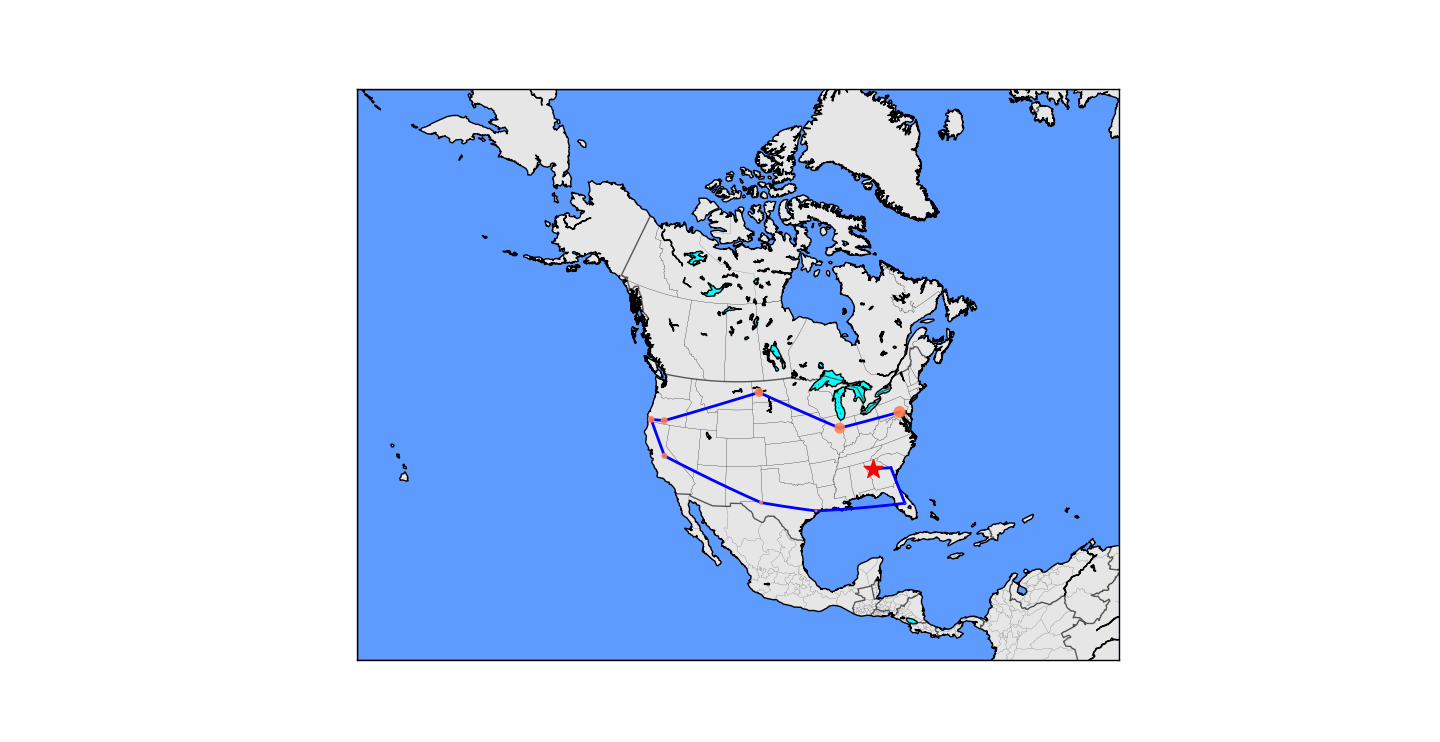}}}%
    \caption{Optimal path-TSP versus TFP routes for 10 U.S.\ locations; the areas of the discs at visit locations are proportional to the expected (quadratic) damage at each location, due to the delay along the route. Here, taking the TFP route (instead of the TSP) reduces total damage by $5\%$.}%
    \label{US-routes}%
\end{figure*}

\paragraph{The Argument versus The Objective.} The set of feasible routes for all $L_p$-TSP problems is the same, while the objectives are different. 
For any $p \neq q$, there exist instances where the two optimal routes are different.
As a simple example, depicted by Figure~\ref{fig:L2vsL1TSP}, consider 
four vertices \textbf{S}, \textbf{A}, \textbf{B}, and \textbf{C} over
a line metric at locations $0,-1-\varepsilon,+1$, and $+2$ respectively. Starting at \textbf{S}, the route \textbf{SABC} is optimal for $L_2$-TSP with corresponding objective of $\|(0,1+\varepsilon, 3+2\varepsilon, 4+2\varepsilon)\|_2 \simeq \sqrt{26}$, in contrast to the optimal route for $L_1$-TSP that is \textbf{SBCA} with the objective $\|(0,1,2,5+\varepsilon)\|_1 = 8 + \varepsilon$.

\smallskip

Consider the effect of a ``wrong'' norm on the optimal route. Figure \ref{US-routes} demonstrates an example in which we consider $10$ locations with simultaneous wildfires in the United States.\footnote{Wildfires observed during the first week of Dec $2019$, according to satellite data \cite{NASA}.}
Dispatching a firefighter from Atlanta, and traveling $100$ times faster than the spread of fire, TFP would look quite different than $L_\infty$-TSP, i.e., path-TSP. In particular, a TFP optimal route reduces the total damage by $5\%$ compared to that of the path-TSP.

As another example, consider destinations on the $x$ axis, with many fires located at $+1$ and a single fire at $-1+\varepsilon$. Starting at $x = 0$, the optimal $L_\infty$-TSP route moves left first and then right, resulting in a solution roughly three times as expensive as the optimal TFP route, in terms of the $L_2$ objective. 
In fact, the relative performance of an $L_\infty$-TSP route for the $L_2$ objective can be \emph{unbounded}. For instance, consider the Euclidean metric over the plane with $n-2$ fires located at complex coded locations $$\{e^{2\pi \cdot \frac{k}{n} i} : k \in \{1, \cdots, n-2\}\},$$ and $m$ distinct fires located at $e^{2\pi \frac{n-1-\varepsilon}{n} i}\,.$ Starting at $(1,0) = e^{0i}$ and moving at unit speed, for $n \rightarrow \infty$ and $m/n \rightarrow \infty$ the damage (squared delay) due to $L_\infty$-TSP route (by walking in the wrong direction around the circle) converges to $4\pi^2$, while for the optimal TFP solution it goes to zero.

\medskip

\paragraph{All-norm-TSP.} 
We observed that TSP may not be a good solution for other $L_p$ objectives. One may consider whether a different $L_p$-TSP is hopefully good for all norms. In this line, a natural question is whether there exists a single route that is approximately optimal concerning the minimization of any norm of the visit times, and whether we can efficiently find one.
This can be viewed also as an online problem where the adversary chooses the norm, e.g.\ $L_p$, and the objective is to provide a competitive solution with respect to the optimal route.

\medskip

\begin{definition}[all-norm-TSP]
Given $s, V \ni s, d: V \times V \rightarrow \mathbb{R}_{\ge 0}$ as before, the objective is to choose a route that minimizes the maximum possible ratio between a symmetric norm of the visit time vector of the output route $\sigma$ and the optimal route for that norm,
$$
\min_{\sigma \in \mathcal{F}} \sup_{\| \cdot \|} \frac{ \lVert \ell_{\sigma} \rVert}{\min_{\sigma' \in \mathcal{F}} \lVert\ell_{\sigma'}\rVert}\,.
$$
\end{definition}

Golovin et al.\ \cite{GGKT08} introduced this problem as the \emph{all-norm-TSP} and gave an algorithm that outputs a route that is a $16$-approximation with respect to any norm of the visit time vector. The concept of all-norm minimization has been of interest in many applications, e.g., for routing, load balancing \cite{KRT01}, and machine scheduling \cite{AERW04,BP03}.

\subsection{Main Results \& Proof Ideas.}
Optimizing the non-linear objective of $L_p$-TSP can be computationally challenging. The problem is already NP-hard even in the linear case, $p = 1$, on a tree metric \cite{S02}, i.e.\ when the metric is pairwise distances over a graph on $V$ that forms a tree. In contrast, TSP on trees is solvable, in linear time. 

Archer and Williamson \cite{AW03} showed that a $(1+\varepsilon)$-approximate solution for TRP exists that is a concatenation of $O(\log n \cdot \varepsilon^{-1})$ TSP paths. This can be generalized as the following Lemma, which enables a quasipolynomial time approximation scheme for weighted trees.
\begin{lemma}[\cite{AW03}]
$L_p$-TSP can be $(1+\varepsilon)$-approximated by a concatenation of $O(\log n \cdot \varepsilon^{-1})$ TSP paths.
\end{lemma}
\begin{proof}
Without loss of generality we can assume integral and polynomially bounded input distances, $d(i,j) \in \set{0} \cup [O(n^2/\varepsilon)] = \set{0} \cup [\tilde{O}(n^2)]$,  as an appropriate quantization by rounding only adds a multiplicative error of $1 + O(\varepsilon/n^2)$ to any edge and keeps any norm of a valid tour, within a factor $(1 \pm \varepsilon)$. 

\smallskip

Now break the optimal route according to time spots $1, (1+\varepsilon), \cdots, (1+\varepsilon)^\gamma$ where $\gamma = O(\log n \cdot \varepsilon^{-1})$ because $O(n^3 \cdot \varepsilon^{-1})$ is a bound on the length of the optimal route. Replacing each sub-route between two consecutive time spots, with the shortest path TSP over the same points does not increase any visit time beyond a factor $1+\varepsilon$, and hence preserves any norm of the visit times by a factor of $1+\varepsilon$.
\end{proof}

The above approach can only lead to a pseudo-polynomial time approximation algorithm, even for trees. 
Reducing the problem to many shortest paths cannot lead to an efficient polynomial-time reduction, because a concatenation of $o(\log n)$ path-TSP routes cannot approximate $L_1$-TSP within a constant factor \cite{S14}.
To further reduce the number of paths, we use the notion of segmented-TSP, introduced by Sitters \cite{S14}, to enable some dependence between consecutive deadlines. This problem requires a (sequence of monotonically non-decreasing) number of destinations to be visited by a number of deadlines, formulated as follows.

\begin{definition}[segmented-TSP]
Given $V \ni s, d: V \times V \rightarrow \mathbb{R}_{\ge 0}$ as before, in addition to integer numbers $n_1 \leq n_2 \leq \cdots \leq n_k \leq |V| = n$ and fractional numbers $t_1 \leq t_2 \leq \cdots \leq t_k$ as inputs, the segmented-TSP problem is a decision problem to verify whether a route exists that visits at least $n_i$ distinct vertices by time  $t_i$ for all $i \in [k]$, starting at $s$. 
\end{definition}

Approximation of segmented-TSP, i.e.\ a decision problem, can be defined as follows.

\begin{definition}
An $\alpha$-approximate solution to a segmented-TSP instance, must visit the first $n_i$ vertices by the modified deadline $\alpha \cdot t_i$, $\forall i \in [k]$, if an answer to the original segmented-TSP exists.
\end{definition}

We generalize a main result of Sitters \cite{S14} that showed TRP can be reduced to (a polynomially many number of approximate) segmented-TSP problems with a constant number of deadlines.

\begin{restatable}{thm}{reductiontheorem}[$L_p$-TSP $\Rightarrow$ segmented-TSP]
\label{thm:reduction}
Let $\varepsilon > 0$ be a constant and $\mathcal{A}$ be an $\alpha$-approximation algorithm for segmented-TSP for some $k = O(1+\varepsilon^{-2})$ of our choice. There is a $(1+\varepsilon) \cdot \alpha$-approximation algorithm for $L_p$-TSP that calls $\mathcal{A}$ (on the same network $(V,s,d)$) for a strongly polynomial number of times.
\end{restatable}

\smallskip

Theorem~\ref{thm:reduction}, proved in Section~\ref{sec:complx}, along with a PTAS for segmented-TSP \cite{S14} on tree metrics and Euclidean metrics imply the following results.

\begin{restatable}{corollary}{PTAScorollary}\label{thm:PTAS}
There exist polynomial-time approximation schemes (PTAS) for $L_p$-TSP on weighted trees and Euclidean plane metrics.
\end{restatable}

Note that due to our reduction of $L_p$-TSP to segmented-TSP, any approximation results for segmented-TSP, on specific metrics, would follow for $L_p$-TSP, at the cost of an additional factor of $(1+\varepsilon)$ to the approximation bound. Our results can be generalized to the case of multiple travelers, starting from arbitrary locations, as discussed in Section~\ref{sec:multi}.

\smallskip

For general metrics, a PTAS is unlikely, as the problem becomes Max-SNP-hard. On the other hand, the constant factor approximability of $L_p$-TSP for general metrics is immediate due to the $16$-approximation of the all-norm-TSP by Golovin et.\ al.\ \cite{GGKT08}. In this line, we improve the approximation bound for all-norm-TSP by a factor of $2$.

\begin{restatable}{thm}{allnormapproximation}\label{thm:8apx}[all-norm-TSP]
There is a polynomial-time algorithm to find a route that is $8$-approximate with respect to the minimization of any symmetric norm of the visit times (including $L_p$-TSP, TFP, and TRP).
\end{restatable}

Theorem~\ref{thm:8apx} is proved in Section~\ref{sec:all-norm}.
Our algorithm builds on the partial covering idea, presented as Algorithm \ref{meta-algorithm}, that was pioneered by Blum et al.\ \cite{BCCPRS94} for TRP, and was developed through subsequent studies \cite{GK98, CGRT03, GGKT08,bienkowski2021traveling}.

\makeatletter
\def\BState{\State\hskip-\ALG@thistlm}
\makeatother

\begin{algorithm}
\caption{Routing via Partial Covering}\label{meta-algorithm}
\begin{algorithmic}[1]
\Procedure{Geometric-Covering}{$V, s, d$}
\State Algorithm Parameters: $b \in (0,\infty),c \in (1,\infty)$ 
\State $i \gets 0$
\While{there remains destinations to visit}
        \State \Comment{Conducting sub-tours}
        \State $C_i \gets$ a maximal route of length $\leq b \cdot c^i$. 
        \State Travel through $C_i$ (and return to the origin)
        \State $i \gets i+1$
\EndWhile
\State \Return an ordering $\sigma$ of $V$ according to their (first) visit time through the above loop.
\EndProcedure
\end{algorithmic}
\end{algorithm}

The parameters $b$ and $c$ strongly affect the performance of the algorithm. For all-norm-TSP, we choose $b = \min_{i,j \in V} d(i, j) $ and $c = 2$, as in \cite{GGKT08}. The key difference of our algorithm is in line $6$. Instead of an approximation algorithm for $k$-TSP, i.e.\ a route of minimum length that visits $k$-vertices, we utilize a milder relaxation, which is a tree (instead of a route/stroll) rooted at $s$, including $k$ vertices, and of total length not larger than an optimal $k$-TSP. Such a \emph{good $k$-tree} can be found in polynomial time using the primal-dual method that solves a Lagrangian relaxation of $k$-TSP \cite{CGRT03,ALW08,PS14}.

On the other hand, we provide a first impossibility result for all-norm-TSP, notably beyond known inapproximability bounds for specific $L_p$-TSP problems. The following is proved in Section~\ref{sub:inappx}.
\begin{restatable}{thm}{allnormTSPlowerbound}
There is no approximation algorithm for all-norm-TSP with multiplicative factor better than $1.78$, independent of $P = NP$ or other complexity hypotheses.
\end{restatable}

The above result reaffirms the need for approximation algorithms specifically designed for each norm. Along this line, we present a randomized $5.65$ approximation algorithm for the Traveling Firefighter Problem.
\begin{restatable}{thm}{TFPapproximation}\label{thm:TFP}
There is a randomized, polynomial-time $5.65$-approximation algorithm for TFP on general metrics.
\end{restatable}

Theorem~\ref{thm:TFP} is proved in Section~\ref{sec:TFP}, for which we adapt the ideas by Chaudhuri et.\ al.\ \cite{CGRT03} and optimize the parameter $c$ in line $2$ of Algorithm \ref{meta-algorithm}. Choosing $b \in [1,c]$ at random, with a distribution of uniform density for $\log(b)$, simplifies the analysis, while one can efficiently de-randomize the algorithm by quantization of $b$.

\subsection{Literature Review.}

Traveling Salesman Problem is a principal problem in computer science, combinatorial optimization, and operations research, and its first formulations date back as early as $19$\textsuperscript{th} century (c.f.\ \cite{ABCC06}). Since the celebrated $3/2$-approximation algorithm of Christofides-Serdyukov \cite{Chr76,Ser78}, for its tour-variant on general metrics, TSP has been extensively studied for half a century \cite{Wol80, SW90, BP91, Goe95, CV00, GLS05, BC11, SWZ12, HNR19}. Very recently, Karlin, Klein, and Oveis Gharan \cite{KKO20} showed TSP can be approximated strictly better than $3/2$, while the problem remains NP-hard to approximate within a factor of $123/122$ \cite{KLS15}.

The common ground in numerous variants of TSP is that a set of vertices are to be visited in the \emph{fastest possible} way, i.e., optimizing the time spent by the server/traveler.
In contrast, 
the Traveling Repairman Problem (TRP), a.k.a.\ Minimum Latency Problem, the school bus driver problem \cite{will1994extremal},
hidden treasure \cite{BCCPRS94,koutsoupias1996searching,ausiello2000salesmen}, and the deliveryman problem \cite{minieka1989delivery,fischetti1993delivery,mendez2008new}, optimizes the route purely from the perspective of clients, i.e., the total waiting time to be visited, and is another extensively studied combinatorial optimization problem \cite{ACPPP86, papadimitriou1993traveling, BCCPRS94, GK98, CGRT03, AW03} with a state-of-the-art approximation factor of $\simeq 3.59$ for general metrics \cite{CGRT03,PS14}.

\smallskip

Containment of fires can be abstracted from various perspectives. Hartnell \cite{hartnell1995firefighter} modeled a constant speed spread of fires through edges of a graph, along which the firefighters also displace. Many objectives, such as minimization of the number of burned vertices, or the required time to contain the fire(s) are studied in this model and the problem is an active area of research. See \cite{finbow2009firefighter, klein2014approximation,adjiashvili2018firefighting,amir2020firefighter,deutch2021multi} and references therein. 

\smallskip

Generalizing the objective to Minkowski norm of the solution has allowed interpolating other classical problems, e.g., $L_p$ set cover problem \cite{GGKT08,BBFT20} that further united the greedy algorithms for set-cover and the minimum-sum-set-cover \cite{FLT04} problems. Set cover is better approximable for $p = 1$ than $p = \infty\,,$ while TSP ($p = \infty$) is currently better approximated than TRP ($p = 1$). Nevertheless, this order is not expected to be reversed as TRP is intrinsically a harder problem.
Moreover, in contrast to the concordance among $L_p$ set cover problems, for which the same greedy algorithm gives best (possible) bounds for any $p \in [1, \infty)$, state-of-the-art algorithms for TSP and TRP are significantly different, and potentially far from the best possible. This is yet another motivation to study $L_p$-TSP, ultimately towards unified best algorithms for all underlying problems.

\smallskip

\subsection{Paper Organization.} In Section~\ref{sec:complx} we provide a reduction from $L_p$-TSP to segmented-TSP and subsequent approximation schemes for Euclidean and weighted-tree metrics. In Section~\ref{sec:all-norm}, we present the $8$-approximation for all-norm-TSP in general metrics, along with a first inapproximability bound. This will also be a preliminary for optimized algorithm for $L_2$-TSP in Section~\ref{sec:TFP}, where we present a $5.65$-approximation for the Traveling Firefighter Problem on general metrics. We discuss generalizations to multiple vehicle scenarios in Section~\ref{sec:multi} and conclude the paper with a set of interesting open problems to continue this line of research.

\section{Reducing $L_p$-TSP to Segmented TSP}\label{sec:complx}

In this section, we provide our reduction from $L_p$-TSP to polynomially many instances of segmented-TSP. Let us restate Theorem \ref{thm:reduction}.

\reductiontheorem*

In particular, a corollary of the above Theorem (and the following Lemma) is a $(1+\varepsilon)$ approximation algorithm for any $L_p$-TSP on weighted-tree metrics -where the problem becomes strongly NP-hard even for $p = 1$ - as well as the Euclidean plane. 

\begin{lemma}[\cite{S14}]
\label{segTSP}
Segmented TSP, for any constant number of segments $M$, can be solved in polynomial time for weighted trees, and $1+\varepsilon$ approximated for unweighted Euclidean metric.
\end{lemma}

\PTAScorollary*

In the rest of this section, we prove Theorem~\ref{thm:reduction}, by providing a dynamic programming algorithm that approximates $L_p$-TSP using polynomially many calls to (approximate) segmented-TSP. More precisely, we show if there is an $\alpha$-approximate solution for segmented-TSP, the dynamic program guarantees an approximation factor of at most $\alpha\cdot(1+\varepsilon)$ for arbitrary constant $\varepsilon > 0\,.$

\smallskip

The algorithm is presented in a few steps, each imposing no more than $1+O(\varepsilon)$ multiplicative error. To achieve exact $1+\varepsilon$ precision, one may run the algorithm for a constant fraction of the target $\varepsilon\,.$

\smallskip

For some $k$ as large as $ O(1+\varepsilon^{-2})$ we can ensure $c \defeq (1+\varepsilon)^k \ge 3\,.$ Let $\Opt^{\lambda_i}$ denote the maximal prefix of $\Opt$ route for $L_p$-TSP of length at most
$$
\lambda_i \defeq (1+\varepsilon)^{-j} \cdot c^i\,, \quad \forall i \ge 0\,,
$$
where $j$ is a fixed random number, uniformly distributed over $\{0, \cdots, k-1\}$.

\smallskip

Let $\Opt'$ be a tour made of sub-tours consisting of traversing $\Opt^{\lambda_i}$ and returning to the origin and waiting until time $3\lambda_i$ before starting the next sub-tour. 
To confirm the above is feasible, we need to show sub-tour $i+1\,,$ being allowed to begin at $3\lambda_i\,,$ does not leave before the return of previous sub-tour, i.e.,
$$
3\lambda_{i-1} + 2 \|T^{\Opt^{\lambda_i}}\|_\infty \leq \lambda_i + 2\lambda_i = 3\lambda_i
$$
which is immediate having $\lambda_i = c \lambda_{i-1} \ge 3\lambda_{i-1}$.

\smallskip

\begin{lemma}\label{lem:loss}
The modified tour $\Opt'$ is approximately optimal in expectation, i.e.,  for some $k \in O(1+\varepsilon^{-2})$ $$\expec{j}{\|T^{\Opt'}\|_p^p} \leq (1+\varepsilon) \|T^\Opt\|_p^p\,.$$
\end{lemma}

\begin{proof}
All vertices are visited (for the first time) in the same order, by $\Opt$ and $\Opt'\,.$
Let the $d$\textsuperscript{th} service time by the optimal solution be $$T_d^\Opt \in ((1+\varepsilon)^\delta, (1+\varepsilon)^{\delta+1}]$$ for some integer $\delta \ge 0\,.$

\smallskip

If this vertex is visited in the $i$\textsuperscript{th} sub-tour of $\Opt'\,,$ we can write
$$
T_d^{\Opt'} =  T_d^{\Opt} + 3\lambda_{i-1}\,.
$$
We can bound this additional delay by 
$$T_d^{\Opt'} - T_d^{\Opt} \leq 3(1+\varepsilon)^{\delta-j'}$$ where $j'$ has the same distribution as $j\,.$

We can prove the desired by bounding the per-vertex ratio by
$$
\frac{\expec{j'}{(T_d^{\Opt'})^p}}{(T_d^\Opt)^p}
 \leq (1+\varepsilon)
$$
because
$
\frac{\sum_{i} a_i}{\sum_{i} b_i} \leq \max_{i} \frac{a_i}{b_i}
$
where $a_1, \cdots$ and $b_1, \cdots$ are positive real numbers.

Considering $p \ge 1$ and $T_d^{\Opt'} = T_d^{\Opt} + 3\lambda_{i-1}$, the left hand side of the target ratio is maximized for $T_d^{\Opt} = (1+\varepsilon)^\delta$, so it suffices to prove
$$
\frac{\expec{j'}{((1+\varepsilon)^\delta + 3\lambda_{i-1})^p}}{((1+\varepsilon)^\delta)^p}
 \leq (1+\varepsilon)\,.
$$
We will have
\begin{align*}
& \frac{\expec{j'}{((1+\varepsilon)^\delta + 3\lambda_{i-1})^p}}{((1+\varepsilon)^\delta)^p} \\
&\leq
\frac{\expec{j'}{((1+\varepsilon)^\delta + 3(1+\varepsilon)^{\delta-j'})^p}}{((1+\varepsilon)^\delta)^p} \\
&=
{\expec{j'}{(1 + 3(1+\varepsilon)^{-j'})^p}} \\
&=
1 + {\expec{j'}{(1 + 3(1+\varepsilon)^{-j'})^p- 1^p}} \\
&\leq
1 + \frac{1}{k} \sum_{j' = 0}^{k-1} (3p)^p (1+\varepsilon)^{-j'} \\
&\leq
1 + \frac{(3p)^p}{k} \cdot \frac{1}{1-(1+\varepsilon)^{-1}} \\
&=
1 + \frac{(3p)^p}{k} \cdot
\frac{1+\varepsilon}{\varepsilon}\,.
\end{align*}

To get the desired (from the last inequality) it suffices to assume
$$
k \ge \frac{(3p)^p(1+\varepsilon)}{\varepsilon^2}\,.
$$
\end{proof}

We can now complete the proof of Theorem \ref{thm:reduction}, similar to the main idea of Sitters \cite{S14}, presented as follows.

\smallskip

Lemma \ref{lem:loss} implies for some $j \in [k]$, where $k = O(1+\varepsilon^{-2})$, there exists a near optimal routing, $\Opt'\,,$ that for each 
$i \in [\tilde{O}(n^2)]$, visits \emph{new vertices} only during $[3\lambda_{i-1},\lambda_i]\,,$ and returns to the origin and remains there until $$3\lambda_i = \frac{3}{(1+\varepsilon)^j} \cdot (1+\varepsilon)^{ki}\,.$$

\smallskip

We can search for such a path by reconstructing $\Opt^{\lambda_i}$ for all $i$ and upper bounding the consequent $\|T^{\Opt'}\|_p^p$ using dynamic programming.

\smallskip

Define $D[i][d]$ as (an upper bound on) the contribution of visit times of vertices that are visited by $\Opt^{\lambda_i}$, to $\|T^{\Opt'}\|_p^p$, further assuming the number of these vertices is $d$.

\smallskip

We can compute $D[i][d]$ considering $O(n^k)$ cases of $(m_1, m_2, \cdots, m_k)$ where $m_r$ denotes the number of vertices that are visited by $\Opt^{\lambda_i}$ during
\begin{align*}
&(3\lambda_{i-1} + \lambda_i \cdot (1+\varepsilon)^{r-k-1}, 
3\lambda_{i-1} + \lambda_i \cdot (1+\varepsilon)^{r-k}]\,. 
\end{align*}
Note that it is necessary to have $\sum_r m_r \leq d$ and let $d' = d - \sum_r m_r$ be the number of vertices visited by $\Opt^{\lambda_{i-1}}$.
We can write
\begin{align*}
&D[i][d] = \min_{m_1,\cdots,m_k} D[i-1][d'] + \\ &\text{Seg-TSP}_i(d',m_{[r]}) \cdot \sum_r m_r \cdot (3\lambda_{i-1} + \lambda_i \cdot (1+\varepsilon)^{r-k})^p\,,
\end{align*}
where $\text{Seg-TSP}_i(d',m_{[r]})$ has value $1$ if segmented-TSP is feasible for visiting at least $$d',d'+m_1,\cdots,d'+m_1+\cdots+m_r$$ vertices by deadlines $$\lambda_{i-1},3\lambda_{i-1}+\lambda_i \cdot (1+\varepsilon)^{-k}, \cdots, 3\lambda_{i-1} + \lambda_i$$ is feasible, and otherwise has value $\infty$. Note that we have an $\alpha$ approximate solver for Segmented-TSP, though for convenience we can alternatively assume the traveller goes at the speed of $\alpha$ instead of $1$, to get a $1+\varepsilon$ approximate solution to $\|T^{\Opt}\|_p$ by $(D[\tilde{O}(n^2)][n])^{1/p}$. In the end, moving at unit speed (instead of $\alpha$) at every stage can increase (any) norm of the delay vector $\|T\|$ by a factor $\alpha$ so we have an $\alpha \cdot (1+\varepsilon)$ approximation, that was promised by Theorem~\ref{thm:reduction}.

\smallskip

Finally it is worth to mention that the route (instead of the value) can be reconstructed using update (parent) information of $D[\cdot][\cdot]$ and a constructive approximate solver for $\text{Seg-TSP}[\cdot]$ and we can short-cut potential re-visits of vertices to have a valid Hamiltonian route.

\section{ All-norm TSP}\label{sec:all-norm}

Since the introduction of a first constant approximation for TRP by Blum et al.\ \cite{BCCPRS94}, partial covering through applying a geometric series of limits on the length of the sub-tours has been a core in the design of routing algorithms.
In this section, we improve the $16$-approximate/competitive solution for all-norm TSP by a factor of $2$. We further provide a first lower bound for this problem.

\subsection{Approximation for General Metrics.}

The idea is to iteratively cover more and more vertices by sub-tours of exponentially (geometrically) increasing length while trying to maximize the total number of vertices that are visited (not necessarily for the first time) in each iteration.
Our algorithm 
uses the following milder relaxation of $k$-TSP, called a \emph{good $k$-tree}, in place of line $6$ in Algorithm \ref{meta-algorithm}.

\begin{definition}
A good $k$-tree is a tree of size $k$, including $s$, and with a total edge-weight of no more than that of the optimal $k$-TSP (starting from $s$).
\end{definition}

\begin{lemma}[\cite{CGRT03}]
A good $k$-tree can be found in polynomial time.
\end{lemma}

Chaudhuri et.\ al.\ \cite{CGRT03} proved the above using a primal-dual approach \cite{garg19963,AK00} that allows finding a feasible solution to the primal (integer) linear program of the $k$-tree problem paired with a feasible dual solution to $k$-TSP, that by weak duality has no less of a cost.

We are now ready to prove Theorem \ref{thm:8apx}.

\allnormapproximation*

\begin{proof}
WLOG assume the nearest neighbor to $s$ is at distance $1$. For $k = 1, 2, \dots, n$ find a good-$k$-tree. Among these, name the largest tree (with respect to number of vertices) of total length at most $2^i$ as $G_i$ for $i = 0, 1, 2, \dots\,.$

\smallskip

Let $C_i$ be a (randomized) depth-first traversal of $G_i$ and let $C$ be the concatenation of $C_i$'s for $i = 0, \dots\,.$ The final tour $\Alg$ will visit vertices in the order that they appear in $C$, which does not increase the first-visit time for any vertex, due to triangle inequality of the metric, while short-cutting vertices that are being re-visited. 

\smallskip

Let $$T_k^\Opt \in [2^i, 2^{i+1})\,.$$ This shows the shortest (length) $k$-path in $G$ is no longer than $2^{i+1}\,.$ So our good-$k$-tree is no longer than $2^{i+1}\,,$ hence $C_{i+1}$ has at least $k$ distinct vertices, allowing us to upper bound our $k$\textsuperscript{th} visit time by
$$T_{k}^\Alg \leq \sum_{j = 0}^{i+1} |C_j| \leq \sum_{j = 0}^{i+1} 2 \times 2^{j} < 2^{i+3}\,.$$

Together with $T_k^\Opt \ge 2^i$ and the above inequality we have
$$
T_k^\Alg \leq 8 \times T_k^\Opt\,.
$$

We showed
$T_k^\Alg \leq 8 \cdot T_k^\Opt \quad \forall i \in [n]$, i.e., $T^\Opt$ is $8$-submajorized by $T^\Alg$ in terminology of \cite{GGKT08,HLP88}. The result is that $$\|T^\Alg\| \leq 8 \cdot \|T^\Opt\|$$ w.r.t.\ any norm $\|\cdot\|$.
\end{proof}

One can verify the above algorithm performs asymptotically $3$ times worse than the optimal TRP for the example with service points at $\{2^i : i \in \mathbb{N}\}$ and starting at $x = 0\,.$

\subsection{Inapproximability.}\label{sub:inappx}
We conclude this section by providing a lower bound for all-norm TSP. We show even for line metrics, an $\alpha$-approximate all-norm TSP cannot be guaranteed in general, for $\alpha < 1.78$.

\allnormTSPlowerbound*

\begin{proof}
We prove this for the special case of a line metric, i.e., when distances are absolute differences between points (vertices) on the real line. Our example has a similar structure as follows. Starting the walk from the origin at $x = 0$, there is a single destination at $x = -1$, in addition to $n$ destinations at $x \in \set{b^i-1 : i \in [n]}$, for $b = 1 + \varepsilon$.
The approximation ratio of the route that first visits points to the right of the origin with respect to $L_{\infty}$-TSP objective is
$
\frac{2b^{n}-1}{b^{n}+1}
$,
that converges to $2$ when $n \rightarrow \infty$.
Alternatively, the approximation ratio w.r.t. $L_1$-TSP for the route that first goes left, i.e., optimal $L_{\infty}$ route, is
$
\frac{1+2n+b^{n+1}/(b-1)-n-1}{b^{n+1}/(b-1)-n-1+2b^n-1}\,.
$
The minimum of these two ratios will be at least $1.67$, achieved for $n = 2100$, $\varepsilon = 1e-3$. Constructing a numerical example with similar structure, depicted as Figure~\ref{fig-example}, we considered all candidate optimal routes the measured approximation ratios with respect to various norms. The $\min \max$ over these ratios was $1.78$, verifying nonexistence of an approximate all-norm-TSP with better performance. For reproducibility we include this example in the Appendix. 
\end{proof}

\begin{figure*}[h]
    \centering
    {\includegraphics[width=15cm]{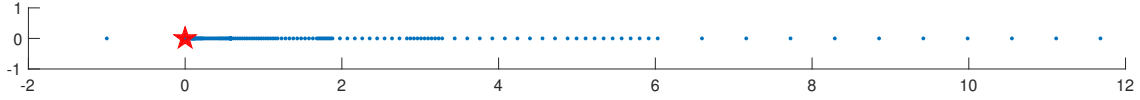}}
    \caption{An example with no better than $1.78$ all-norm TSP}%
    \label{fig-example}%
\end{figure*}

The above example, is yet another motivation to study and optimize routing algorithms specific to the appropriate objective/norm, as one solution \emph{cannot} be good for all.

\section{Traveling firefighter in general metrics}\label{sec:TFP}

In this section we build upon our geometric partial-covering algorithm with good $k$-trees to improve approximation bound for a specific norm, i.e., the Traveling Firefighter Problem.

\smallskip

We achieve the improved approximation bound by randomization (of parameter $b$) and optimization of our analysis w.r.t.\ approximation guarantee for a specific norm. Our approach can provide approximation guarantees (better than $8$) for other $L_p$-TSP problems. We present main ideas in the rest of the section by proving the following result.

\begin{theorem}\label{thm:TFP2}
Traveling Firefighter Problem for general metrics can be $5.641$-approximated in polynomial time.
\end{theorem}

\begin{proof}
We present a randomized approximation algorithm for general metrics, that can be efficiently de-randomized.

\smallskip

Among the set of good $k$-trees, pre-computed for all $k\,,$ let $G_i$ be the largest one of total length at most $b\cdot c^i\,,$ and let $C_i$ be a depth first traversal of that. Parameter $c > 1$ is a constant, to be optimized for performance guarantees, and $$[1,c] \ni b = c^U$$ where $U$ is a random variable distributed uniformly over the interval $[0,1]\,.$ Finally, we reverse each $C_i$ with probability half, and concatenate $C_i$'s (and shortcut repeated visits) to achieve the output ordering $\Alg\,.$

\smallskip 

Let the latency of the $k$\textsuperscript{th} vertex visited by the optimal route be $$T_k^\Opt = a c^i\,,$$ for some $a \in [1, c]$ and integer $i \ge 0\,.$ 

\smallskip

It is easy to see that our $(i + \mathbbm{1}[a \ge b])$\textsuperscript{th} sub-tour contains at least $k$ vertices, hence, we can bound
$$
T_k^\Alg \leq X_k + \sum_{j = 0}^{i - \mathbbm{1}[a < b]} 2 b c^j\,,
$$
Where $X_k$ is zero if our tour visits its $k$\textsuperscript{th} vertex before the $(i - \mathbbm{1}[{a < b}])$\textsuperscript{th} sub-tour, and $$X_k \in [0, 2bc^{i + \mathbbm{1}[a \ge b]}]\,,$$ depending on its location in the sub-tour.

\smallskip

Due to convexity of the quadratic function we can bound the expected damage, with
\begin{align*}
&\expec{}{(T_k^\Alg)^2} \\&\leq \frac{1}{2} \left( \sum_{j = 0}^{i - \mathbbm{1}[a < b]} 2 b c^j \right)^2 \\
&+ \frac{1}{2} \left(2bc^{i+1 - \mathbbm{1}[a < b]} + \sum_{j = 0}^{i - \mathbbm{1}[a < b]} 2 b c^j\right)^2 \\
&= \frac{1}{2}\left( 2b \frac{c^{i + \mathbbm{1}[a \ge b]}-1}{c-1} \right)^2 \\
&+ \frac{1}{2}\left( 2b \frac{c^{i+1+ \mathbbm{1}[a \ge b]}-1}{c-1} \right)^2\\
&\leq \frac{1}{2} \Bigl(\frac{2bc}{c-1}\Bigr)^2 c^{2(i+\mathbbm{1}[a \ge b])} \\
&= c^{2i} \cdot \frac{2c^2}{(c-1)^2} \Bigl(b^2c^{2 \cdot \mathbbm{1}[a \ge b]}\Bigr)\,.
\end{align*}

Bounding the expected damage at the $i$\textsuperscript{th} service, considering random variable $b$ we have
\begin{align*}
&\expec{}{(T_k^\Alg)^2} \\ &\leq \expec{}{c^{2i} \cdot \frac{2c^2}{(c-1)^2} \left(b^2c^{\mathbbm{2}[a \ge b]}\right)} \\
&= c^{2i} \cdot \frac{2c^2}{(c-1)^2} \left( c^2\int_0^{\log_c{a}} c^{2U} dU  + \int_{\log_c{a}}^1 c^{2U} dU \right) \\
&= c^{2i} \frac{2c^2}{(c-1)^2} \left( \frac{c^2(a^2-1)}{\ln c} + \frac{c^2 - a^2}{\ln c}\right) \\
&= (ac^i)^2 \left(\frac{2c^2 \cdot (c^2-1)}{(c-1)^2 \ln c}\right) \\
&= (T_k^O)^2 \left(\frac{2c^2(c+1)}{(c-1)\ln c}\right)\,.
\end{align*}

We can now choose $c$ in order to minimize the multiplicative bound
$$
\frac{c+1}{c-1} \cdot 2c^2 / \ln c \leq 31.82
$$
that can be achieved for  $c \simeq 2.54\,.$ With this we have
$$
\expec{}{\|T^\Alg\|_2^2} \leq 31.82 \|T^\Opt\|_2^2\,.
$$

We provided a randomized algorithm, that along law of large numbers, can be applied in practice, by repeating the algorithm and reporting the best performing route. On the other hand it can be simply de-randomized by exploring all values for a dense enough quantization of $b\,.$ In the first case we will have a $\sqrt{31.82} \simeq 5.641$ approximate $L_2$-TSP with high probability, and in latter scenario we will have deterministic result with negligible additional approximation error.
\end{proof}

\section{Generalizations to Multiple Vehicles}\label{sec:multi}

In many applications we have multiple vehicles, potentially dispatching from different hubs \cite{kulichmulti20}. 

\smallskip

Our reduction of $L_p$-TSP to Segmented-TSP can be generalized to multiple vehicles, with arbitrary start locations. Similar to the approach in Section~\ref{sec:complx} we can convert any optimal multi-vehicle solution to repetition of $O(\varepsilon^{-2})$ of prefix routes for each vehicle, all synchronized with a single $j \in \{0, \cdots, k-1\}$, picked uniformly at random. Lemma~\ref{lem:loss} can be subsequently adapted to allow limiting the search space to solutions that have all vehicles at starting locations simultaneously at all $3\lambda_i$, with negligible degrade of the optima.

\smallskip

Finally, adapting the dynamic programming, we can guarantee a multiplicative $O(\varepsilon)$ loss given an (approximate) solver for multi-vehicle segmented TSP with constant $O(\varepsilon^{-2})$ deadlines, and the corresponding total number of destinations to be visited (by at least one vehicle) until up to each. This is indeed fruitful as results on segmented-TSP also generalize to multi-vehicle variant, e.g., in tree-distance or Euclidean metric.

\smallskip

Our algorithms can be similarly adapted in the case where release dates are added for the destinations. 

\smallskip

Generalizing results in Sections~\ref{sec:all-norm}-\ref{sec:TFP} to multi-vehicle variants of the problems is also possible. For this purpose, the $k$-stroll subproblem, for which we used the mild relaxation of good $k$-tree, can be generalized to bottleneck-stroll of Post and Swamy \cite{PS14}. This will be at the cost of further degrades to the approximation constants and proposes interesting open problems for study. For general metrics, current best approximation guarantees for multi vehicle $\{\text{single},\text{multi}\}$ depot $L_1$-TSP is $\{7.183,8.497\}$ \cite{PS14}.

\section*{Discussion and Open Problems}

We studied combinatorial optimization problems whose objectives can be more appropriate, efficient, fair, and adjustable, depending on enormous applications of optimal routing/scheduling.

\medskip

For TSP and TRP, the analyses of approximation algorithms as well as complexity results heavily rely on the linearity of the objective function. Hence, TFP and more generally $L_p$-TSP pose further challenging problems and require new techniques to be developed.

\medskip

We developed two approaches, towards high precision and scalable approximation of the answer.

\medskip

First, we provided a high precision polynomial time reduction of $L_p$-TSP to segmented-TSP with only a constant number of deadlines for visiting the required number of destinations.
Our reduction enables approximation schemes for $L_p$-TSP on Euclidean as well as weighted tree metrics; this is yet another motivation to further study the segmented-TSP problem.

\medskip

The other approach relied mainly on the fact that the objective is a norm of the delay vector. In this line, we developed an algorithm for all-norm-TSP, on general metrics, of approximation factor $8$. We also provided a first inapproximability result for all-norm-TSP. 

\medskip

Last but not least, we showed how the performance of the latter algorithm can be optimized for a specific norm, particularly approximating Traveling Firefighter Problem within a factor $5.65$.

\medskip

In the end, in addition to further improving the approximation bounds for the problems under study, we mention but a few of many interesting open problems and potential research directions regarding $L_p$-TSP and all-norm-TSP.

\begin{itemize}
    \item $L_1$-TSP, i.e., TRP is harder than $L_\infty$-TSP, at least on trees. For what $p$ is the $L_p$-TSP problem the hardest?
    \item While TRP is strongly NP-hard on weighted trees, its complexity is unknown for caterpillars \cite{S02}. Similarly TFP and $L_p$-TSP seem challenging even on such fundamental examples, that is yet to be resolved.
    \item While we theoretically claimed $p = 2$ to be ideal for the Traveling Firefighter Problem, 
    this assumption should be given further justification / investigation  in practice.
    \item Further applications of $L_p$-TSP can be inspected, e.g.\ in optimal containment of spread of pandemics \cite{hartke2004attempting,tennenholtz2020sequential}.
    \item We observe $c = 2$ to be optimal for the analysis of all-norm-TSP algorithm. For TRP, the current best result is due to a base $c \approx 3.59$ for the geometric series, while $c \approx 2.54$ is better for TFP, as we discussed. In this vein, one can inspect the best base for the geometric series used by the partial covering algorithm, depending on the objective, which naturally seems to be non-increasing on $p$. 

    \item Stronger impossibility results for all-norm-TSP and even larger hardness of approximation bounds for this problem seem plausible.
\end{itemize}

\section*{Acknowledgements}
We would like to thank anonymous reviewers for their comments.
First author would like to thank Joseph Bakhtiar for discussions, and Jai Moondra for a careful reading of the paper, and their helpful comments.

\bibliography{refs}
\bibliographystyle{amsalpha}

\newpage

\appendix

\section{From Section \ref{sec:all-norm}}
Our example for $1.78$ inapproximability of all-norm-TSP has similar structure as the exponential sequence presented in Section \ref{sec:all-norm}, yet is computationally tuned and verified. This suggests the best (worst) example can have a different structure. Follows the locations of the destinations over the line to be visited by a traveler, starting at $x = 200$. All points are over the $x$-axis. Figure~\ref{fig:last} depicts performance of candidate routes with respect to different norms.

\begin{align*}
    V = \{ 0, 200, 202, 204, 206, 208, 210, \\212, 214, 216, 217, 218, 219, 220, 221, \\222, 223, 224, 225, 226, 228, 230, 232, \\234, 236, 238, 240, 242, 244, 246, 250, \\254, 258, 262, 266, 270, 274, 278, 282, \\286, 289, 292, 295, 298, 301, 304, 307, \\310, 313, 316, 316, 316, 316, 316, 316, \\316, 316, 316, 316, 316, 322, 328, 334, \\340, 346, 352, 358, 364, 370, 376, 382, \\388, 394, 400, 406, 412, 418, 424, 430, \\436, 446, 456, 466, 476, 486, 496, 506, \\516, 526, 536, 540, 544, 548, 552, 556, \\560, 564, 568, 572, 576, 595, 614, 633, \\652, 671, 690, 709, 728, 747, 766, 775, \\784, 793, 802, 811, 820, 829, 838, 847, \\856, 888, 920, 952, 984, 1016, 1048, \\1080, 1112, 1144, 1176, 1199, 1222, \\1245, 1268, 1291, 1314, 1337, 1360, \\1383, 1406, 1519, 1632, 1745, 1858, \\1971, 2084, 2197, 2310, 2423, 2536 \}
\end{align*}

\begin{figure}[t]
    \centering
    \includegraphics[width=0.5\textwidth]{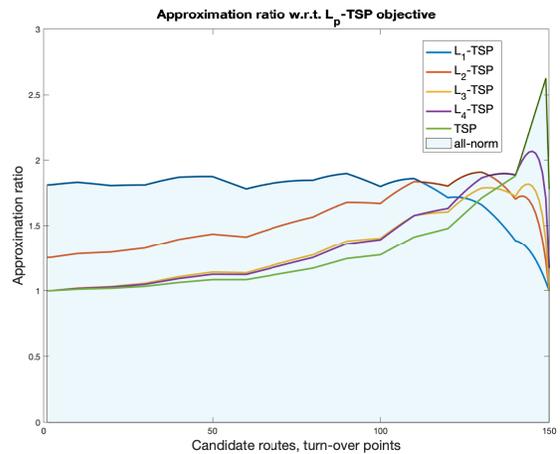}
    \caption{Nonexistence of a $1.78$-approximate all-norm-TSP}
    \label{fig:last}
\end{figure}

\end{document}

\section{OTHER}
In this section we show $L_p$-TSP has no PTAS on general metrics and further reduce its approximation (up to arbitrary small multiplicative error) to segmented-TSP.

\subsection{APX-Hardness}
Let us restate Proposition \ref{prop:APX}. The proof relies on hardness results for Hamiltonian path.

\begin{theorem}
Any $L_p$-TSP is strongly NP-hard. Moreover, it has no PTAS, unless $P = NP\,.$
\end{theorem}
\begin{proof}[Proof idea]
Similar to the approach by Sahni and Gonzalez \cite{ST76}, we can provide a reduction from Hamiltonian path problem. Given a graph $H = (V,E)\,,$ we are to decide whether a Hamiltonian path starting at $s$ exists. We construct a complete graph $G = K_{|V|}$ with the following edge lengths.

\begin{align}
    d(i,j) =
    \begin{cases}
    1, & \text{for } (i,j) \in E\\
    \lambda+1, & \text{otherwise}
    \end{cases}
\end{align}

It is easy to see the visit latency vector for $G$ is $(0,\cdots,n-1)$ iff $H$ admits a Hamiltonian path starting at $s\,.$ Otherwise the largest entry of the vector is increased by at least $\lambda+1\,,$ while neither decreases. The norm of the vector is subject to a constant increase for some $\lambda \in O(n^2)\,.$
\end{proof}

\section{Traveling Firefighter Problem}\label{sec:TFP}

\end{document}